\documentclass[runningheads]{llncs}
\usepackage[T1]{fontenc}
\usepackage{graphicx}
\usepackage{amsmath} 
\usepackage{algorithm}
\usepackage{algpseudocode}
\usepackage{amssymb}
\usepackage{comment}

\newcommand{\etal}{\textit{et al.}}

\begin{document}
\title{Hollow Victory: How Malicious Proposers Exploit Validator Incentives in Optimistic Rollup \\ Dispute Games
}

\author{Suhyeon Lee \inst{1} \inst{2}}

\institute{
Tokamak Network\\
\email{suhyeon@tokamak.network}
\and
Korea University\\
\email{orion-alpha@korea.ac.kr}
}

\maketitle              %
\begin{abstract}
Blockchain systems, such as Ethereum, are increasingly adopting layer-2 scaling solutions to improve transaction throughput and reduce fees. One popular layer-2 approach is the Optimistic Rollup, which relies security on a mechanism known as a dispute game for block proposals. In these systems, validators can challenge blocks that they believe contain errors, and a successful challenge results in the transfer of a portion of the proposer’s deposit as a reward. In this paper, we reveal a structural vulnerability in the mechanism: validators may not be awarded a proper profit despite winning a dispute challenge. We develop a formal game-theoretic model of the dispute game and analyze several scenarios, including cases where the proposer controls some validators and cases where a secondary auction mechanism is deployed to induce additional participation. Our analysis demonstrates that under current designs, the competitive pressure from validators may be insufficient to deter malicious behavior. We find that increased validator competition, paradoxically driven by higher rewards or participation, can allow a malicious proposer to significantly lower their net loss by capturing value through mechanisms like auctions. To address this, we propose countermeasures such as an escrowed reward mechanism and a commit-reveal protocol. Our findings provide critical insights into enhancing the economic security of layer-2 scaling solutions in blockchain networks.

\end{abstract}
\section{Introduction} \label{section: introduction}

Blockchain technology, and Ethereum in particular, has witnessed rapid growth in recent years. However, scalability and high transaction fees remain critical challenges. To address these issues, various layer-2 scaling solutions have been proposed, among which \emph{Optimistic Rollups} have gained significant attention. Optimistic Rollups enable off-chain computation while relying on an on-chain dispute resolution mechanism to ensure correctness. In these systems, block proposals are accepted optimistically, and validators are empowered to challenge any block they believe to be incorrect through a process known as a \emph{dispute game}.

In a typical dispute game, a block proposer stakes a deposit and submits a block. Validators then monitor the block and, if they detect an error, submit a challenge. This challenge is called a fraud proof \footnote{It is also referred to as a fault proof by Optimism, one of the leading optimistic rollup teams. However, the term fraud proof appears to be more widely accepted among optimistic rollups. Therefore, throughout this paper, we will use the term fraud proof.}. If the challenge is successful, the whole or a fraction of the proposer’s deposit is confiscated and awarded to the challenger. Ideally, this mechanism is meant to deter malicious behavior by imposing significant economic penalties on a proposer who submits an incorrect block. Conversely, if the validator loses the challenge, their deposit for the challenge will be forfeited and granted to the block proposer. However, if the incentive mechanism is not carefully calibrated, a malicious proposer might be able to limit his losses to only a small challenge fee even when many validators participate in challenging the block. It results in undermining the intended economic deterrence.

In this paper, we model the dispute game using game-theoretic methods and analyze several attack scenarios with focusing on the malicious block proposer under a fraud proof challenge. We consider a case where the malicious proposer controls a subset of validators as well as cases in which a secondary auction is deployed to induce additional validator competition to get the dispute winner prize. Our analysis reveals that, under current incentive structures, the liveness of the dispute game may be insufficient to motivate them, thereby leaving the system vulnerable to strategic exploitation.

Our contributions are as follows:
\begin{itemize}
    \item We develop a game-theoretic model of the dispute game in Optimistic Rollups, capturing the interactions between a malicious block proposer and validators under various challenge scenarios.
    \item We identify a critical vulnerability: malicious proposers can strategically minimize their own financial losses after being challenged, while simultaneously exploiting validator competition in a way that can diminish the net profitability for challengers, weakening the overall security model.
    \item We propose countermeasures, including an escrowed (deferred) reward mechanism and a commit-reveal protocol to rebalance the incentive structure and enhance the economic security of Optimistic Rollups.
\end{itemize}

The rest of the paper is organized as follows: Section \ref{section: related} reviews related work on dispute games and incentive mechanisms in layer-2 solutions. In Section \ref{section: model}, we present our formal model of the Optimistic Rollup dispute game. Section \ref{section: strategy} details the strategies a malicious proposer employ and introduces our proposed auction mechanism. Section \ref{section: analysis} provides a game-theoretic analysis of the auction mechanism, including conditions under which validators are incentivized to participate. Section \ref{section: solution} proposes two countermeasures to solve this structural vulnerability.
Finally, Section \ref{section: conclusion} discusses potential solutions and future research directions, and concludes the paper.

\section{Related Works} \label{section: related}

A variety of recent research has focused on understanding the economics and security of layer‐2 scaling solutions. In this section, we review them particularly in the context of optimistic rollups. 

\subsubsection{Economic and Incentive Analysis in Rollups.}
Early works in the economic analysis of rollups have focused on establishing robust security guarantees through well-aligned incentive mechanisms. Tas \etal~\cite{tas2022accountable} propose a framework for \emph{accountable safety} in rollups, which emphasizes the need for designs that hold participants economically accountable for misbehavior. Li~\cite{li2023security} further explores the security of optimistic blockchain mechanisms, highlighting the importance of ensuring that validator incentives are strong enough to deter malicious actions. In a similar vein, Mamageishvili and Felten~\cite{mamageishvili2023incentive} analyze incentive schemes for rollup validators, discussing strategies under the attention test.

\subsubsection{Data Availability Cost Optimization.}
Given that Optimistic Rollups are designed as a layer-2 scaling solution, minimizing the costs associated with data availability is a critical concern. Several studies have tackled this problem from different angles. Palakkal \etal~\cite{Palakkal2024sok} provide a systematization of compression techniques in rollups, outlining methods inefficiency in some rollups' practice. Mamageishvili and Felten~\cite{mamageishvili2023efficient} propose efficient rollup batch posting strategies on the base layer using call data, while Crapis \etal~\cite{crapis2023eip} offer an in-depth analysis of EIP-4844 economics and rollup strategies including blob cost sharing. Complementary empirical investigations by Heimbach and Milionis~\cite{Heimbach2025}, Huang \etal~\cite{huang2024two}, and Park \etal~\cite{park2024impact} further document the inefficiencies in current DA cost structures and examine their impact on rollup transaction dynamics and consensus security. Lee~\cite{lee2024180} investigates blob sharing as a potential remedy for the dilemma faced by small rollups in the post-EIP-4844 era.

\subsubsection{Fraud Proof and Dispute Resolution Protocols.}

Prior research on fraud proofs in Optimistic Rollups has primarily focused on ensuring rapid dispute resolution and robust liveness. For example, BoLD~\cite{alvarez2024bold} and Dave~\cite{nehab2024dave} propose protocols that optimize the dispute game to achieve low delays and cost-efficient on-chain verification, while Berger \etal~\cite{berger2025economic} analyze economic censorship dynamics to guarantee that disputes are resolved even under adversarial conditions. While previous studies focus on ensuring the liveness of the dispute game, this paper present study concentrates on a different vulnerability: even when dispute games proceed smoothly, the existing incentive structure may fail to adequately reward honest validators, thereby allowing a malicious block proposer to strategically minimize his costs.

\section{Dispute Game Model for Optimistic Rollups} \label{section: model}

In this section, we present a formal model of the dispute game used in optimistic rollups. We define the game in the style of a game-theoretic construct, denoted by $\mathcal{G}$, and specify the participants, assumptions, rules, and reward mechanisms.

\subsection{Participants and Deposits}
We consider two types of participants in the dispute game:
\begin{itemize}
    \item \textbf{Proposer:} Denoted by $P$, the proposer is responsible for submitting a proposed block.
    \item \textbf{Validators:} Denoted by the set $\mathcal{V} = \{ v_1, v_2, \ldots, v_n \}$, the validators (or challengers) participate in the dispute game to verify that a state transition is incorrect.
\end{itemize}

Each participant is required to hold a minimum deposit to obtain participation rights:
\begin{itemize}
    \item The proposer holds a deposit $D_P$.
    \item Each validator $v_i$ holds a deposit $D_V$, where we assume the same minimum deposit $D_V$ for all validators.
\end{itemize}

\subsection{Rules of the Dispute Game}
The dispute game $\mathcal{G}$ is triggered when one of more validators submit a challenge to a block by proposer $P$. The key rules are as follows:
\begin{enumerate}
    \item \textbf{Verification of State Transitions:} A dispute is resolved by verifying the correctness of the state transition. This is achieved either:
    \begin{enumerate}
        \item via the opponent's timeout, or
        \item through the execution of a bisection protocol on a dispute game contract, wherein the virtual machine (VM) executes the disputed state transition steps.
    \end{enumerate}
    
    \item \textbf{Collateral for Dispute:} Upon initiation of the dispute game, both parties are required to post an additional collateral deposit $D_g$, where
    \[
    0 < D_g < D_{\max}.
    \]
    This additional deposit is intended to further incentivize honest participation in the dispute resolution process.
\end{enumerate}

\subsection{Reward and Penalty Policy}
The reward policy of the dispute game is designed to penalize a malicious proposer and reward the challenger(s) as follows:
\begin{itemize}
    \item If the proposer $P$ loses the dispute game, the challenger(s) may claim a portion of the deposits. Specifically, the winning party is entitled to claim:
    \[
    \alpha \cdot D_P + D_g,
    \]
    where $\alpha$ is a reward parameter.In most real-world rollups, the majority of systems states nearly all or most of the rewards are given to challengers, implying a value close to 1 (100\%). Only one of the leading optimistic rollup teams, Arbitrum's document mentions only 1\% of the confiscated asset will be certainly rewarded to the dispute game winner \cite{arbitrum_economics_of_disputes}. Still, its DAO can reward more of the confiscated asset to the winner. Therefore,
    \[
    0 < \alpha \le 1.
    \]
    \item The deposits that are not forfeited are allocated to a communal fund (e.g., managed by a DAO) to support the overall system.
\end{itemize}

\subsection{Additional Assumptions}
For the purpose of our analysis, we assume:
\begin{enumerate}
    \item \textbf{Sequential and Parallel Occurrence:} The dispute game can occur in multiple rounds sequentially; furthermore, different instances of the dispute game may run in parallel.
    \item \textbf{Timeouts:} Timeouts are set generously so that the proposer cannot deliberately force a timeout by delaying responses. The creation time of each dispute game, denoted by $t_i$, is sufficiently long to preclude strategic timeouts.
\end{enumerate}

This model captures the strategic interaction between the proposer and validators under the assumptions of minimum deposits and additional collateral requirements. In our subsequent analysis, we use game-theoretic frameworks to examine the incentive imbalances that may arise, particularly focusing on how the profit for a validator as a dispute game winner can be limited nearly to zero.

\section{Strategy of Malicious Block Proposer} \label{section: strategy}

In this section, we illustrate the various strategies that a malicious proposer may employ to minimize his cost in the dispute game. We consider three scenarios and then describe an auction mechanism that the malicious proposer can trigger to force additional validator participation. In our analysis, a malicious proposer is denoted by \(P\) and the set of validators by \(\mathcal{V}\).

\subsection{Scenario Descriptions}

We consider the following three scenarios. The first scenario assumes the malicious proposer controls at least one validator. On the other hand, the left two scenarios do not assume the proposer-controlled validator. The difference between the scenario 2 and 3 is that the number of submitted challenges before the malicious proposer responds to any challenge. The scenario 3's feasibility will be analyzed in the next section.

\subsubsection*{Scenario 1 (Proposer-Controlled Validator Accounts)}
In this baseline scenario, the malicious proposer \(P\) controls a subset \(\mathcal{V}_P \subset \mathcal{V}\) of validator accounts. When a block \(B\) is proposed and a valid challenge is submitted by some validator \(v_i \in \mathcal{V} \setminus \mathcal{V}_P\), \(P\) can immediately counter by initiating a dispute game through one of his controlled validators \(v_j \in \mathcal{V}_P\). Under the reward mechanism (where a winning party receives \(\alpha \cdot D_P + D_g\) and \(\alpha \in (0,1]\)), the net cost for \(P\) is limited to:
\[
\text{Cost}_{P} = (1-\alpha) \cdot D_P,
\]
since the controlled validator wins the challenge and the attacker absorbs only the shortfall.

\subsubsection*{Scenario 2 (Multiple Challenges Exist)}
Suppose that multiple validators submit valid challenges against \(B\). In this situation, \(P\) may reduce his effective cost by deploying an auction contract. The auction mechanism is designed to order the challenges: validators who have not yet challenged are induced to bid. Since a validator’s alternative is to obtain a reward of zero (if they do not participate), the auction forces competitive bidding. In equilibrium, validators will bid amounts that reflect their full valuation, roughly given by
\[
v = \alpha \cdot D_P + D_g.
\]
Thus, if \(P\) loses the resulting auction-based challenge, his cost will be at least the winning bid, which in a competitive setting is driven upward by the participation of many validators.

\subsubsection*{Scenario 3 (Single Challenge Initially Submitted)}
In the case where only a single challenge is initially submitted, \(P\) may still deploy the auction contract to attract additional challenges. Even though one challenge exists, the auction mechanism creates a framework in which validators are incentivized to bid—since non-participation yields zero reward. The competitive pressure in the auction will then increase the cost borne by \(P\) if a challenger wins the dispute game. In both Scenarios 2 and 3, the key idea is that by forcing additional validator participation via an auction, \(P\)’s effective cost is raised beyond the minimal challenge fee.

\subsection{Auction Contract Construction}

To explain the scenario in which multiple or a single initial challenge is submitted (Scenarios 2 and 3), we propose an auction-based mechanism that integrates directly with the dispute game. The auction contract is designed to foster competition between validators, thereby minimizing the malicious proposer's loss. In essence, if a valid challenge has already been submitted, the malicious proposer can trigger this auction contract to invite further challenges. As an on-chain smart contract, the auction contract ensures that the auction winner (validator) will win the dispute game first. Otherwise, the auction operator's (malicious proposer) deposit in the auction contract will be confiscated and given to the auction winner.

\begin{algorithm}[htb]
\caption{Dispute Game Auction Contract}
\label{alg:dispute_auction}
\begin{algorithmic}[1]
\State \textbf{Input:} \texttt{disputeId}, auction duration \(T_{\text{auction}}\)
\State \textbf{Initialize:} 
\State \quad \(\texttt{auctionStart} \gets\) current timestamp (or block number)
\State \quad \(\texttt{Bids} \gets \varnothing\)
\State \quad Malicious proposer deposits \(E\) (equal to the promised reward)
\Procedure{SubmitBid}{validator, bidAmount}
    \If {current time \(<\) \(\texttt{auctionStart}+T_{\text{auction}}\)}
        \State Verify dispute game instance exists for \texttt{disputeId}
        \State Append \((\text{validator},\ bidAmount)\) to \(\texttt{Bids}\)
    \EndIf
\EndProcedure
\Procedure{FinalizeAuction}{}
    \If {current time \(\geq\) \(\texttt{auctionStart}+T_{\text{auction}}\)}
        \State Let \(b_{\max}\) and \(b_{\text{second}}\) be the highest and second-highest bids in \(\texttt{Bids}\)
        \State Winner \(\gets\) validator associated with \(b_{\max}\)
        \State Grant Winner the exclusive right to initiate the dispute game at cost \(b_{\text{second}}\)
    \EndIf
\EndProcedure
\Procedure{ResolveDispute}{outcome}
    \If {outcome = first win for the winning validator}
        \State Refund excess deposit to Winner
    \Else
        \State Transfer \(b_{\text{second}}\) from Winner's deposit to malicious proposer
    \EndIf
\EndProcedure
\end{algorithmic}
\end{algorithm}

Our design leverages only minimal EVM functionalities—such as reading block timestamps/number, basic list manipulations, and conditional checks—which makes it well-suited for implementation in a standard smart contract environment. The pseudocode (presented in Algorithm \ref{alg:dispute_auction}) outlines the essential steps:

We now describe the detailed steps of our proposed auction mechanism, which is designed to force additional validator participation and increase the economic cost for a malicious block proposer. The mechanism proceeds as follows:

\textbf{Auction Initialization:} When at least one valid challenge is detected, the auction contract is deployed. At this point, the malicious block proposer transfers, as a deposit, an amount corresponding to the reward promised by the dispute game setting. This initialization triggers the fixed bidding period during which the auction will accept bids.

\textbf{Bid Submission:} Validators who have not yet submitted a challenge are invited to participate in the auction. Prior to bidding, each validator must initiate a dispute game for the challenged block. This step is essential as it enables the auction contract to verify the existence of a corresponding dispute game contract and ensures that all contractual conditions are met. If a validator's valuation (which reflects the full reward) exceeds the current highest bid, they submit their bid along with the necessary deposit.

\textbf{Auction Finalization:} Once a pre-defined number of blocks have passed (or the auction duration has elapsed), the auction period terminates. At this point, the auction contract finalizes the bidding process using a Vickrey (second-price) auction mechanism, whereby the highest bidder wins but pays the amount of the second-highest bid. The winning validator is granted the exclusive guarantee to finish the dispute game challenge first.

\textbf{Dispute Resolution and Cost Recovery:} The dispute game then proceeds, and the system monitors all submitted challenges. Once it is verified—by querying the status of all dispute games—that the winning validator's dispute game is the first to conclude, the auction contract triggers the cost-recovery process. Specifically, the winning validator’s deposit is used to transfer the winning bid amount to the malicious block proposer, while any remaining difference (if the winning bid is less than the full deposit) is refunded to the winning validator.

In summary, this auction contract serves as a critical component of our overall strategy. It transforms the challenge submission process into a competitive auction, where the equilibrium bid is driven by each validator's valuation.

\section{Analysis on the Dispute Game Auction} \label{section: analysis}

In our dispute game auction, each validator faces a fixed participation cost \(c>0\) (accounting for gas fees, contract invocation, etc.). The reward available if a validator wins a dispute is given by
\[
R = \alpha D_P + D_g,
\]
where \(D_P\) is the forfeitable deposit of the malicious proposer, \(D_g\) is the additional collateral for the dispute game process, and \(\alpha\in(0,1]\) is a reward parameter. In practice, since the protocol is uniform, validators’ valuations are very similar. To capture this slight heterogeneity, we assume that each validator’s valuation is drawn from the interval \([R-\mu, R]\), where \(\mu>0\). Under these assumptions, the potential surplus a validator can obtain by winning the auction comes solely from the dispersion \(\mu\).

\subsection{Expected Payoffs and Participation Conditions in a Second-Price Auction}

Assume that \(n\) validators participate in a second-price (Vickrey) auction created by the malicious block proposer. Since the valuations are nearly identical, each validator’s true valuation is \(v_i\in [R-\mu, R]\). By a change of variable, let
\[
X_i = v_i - (R-\mu),
\]
so that \(X_i\) is uniformly distributed on \([0,\mu]\). It is well known that in a second-price auction with \(n\) independent draws from \(U[0,\mu]\), the expected maximum is
\[
E[X_{(1)}] = \frac{n}{n+1}\mu,
\]
and the expected second-highest value is
\[
E[X_{(2)}] = \frac{n-1}{n+1}\mu.
\]
Thus, the expected surplus (i.e., the difference between the highest and second-highest bid) is
\[
E[X_{(1)} - X_{(2)}] = \frac{\mu}{n+1}.
\]
Since each validator pays a fixed cost \(c\) upon participation regardless of winning or losing, the net gain for a winning validator is
\[
\pi_{\text{win, validator}} = \frac{\mu}{n+1} - c.
\]
A validator who loses simply incurs a loss of \(c\). Therefore, a validator’s decision to participate (as opposed to abstaining and receiving a payoff of zero) is individually rational if the expected surplus exceeds the cost:
\[
\frac{\mu}{n+1} - c > 0 \quad \Longleftrightarrow \quad \mu > c\,(n+1).
\]
We formalize this result as follows.

\begin{theorem}
Let \(n\) risk-neutral, symmetric validators have valuations drawn independently from the interval \([R-\mu, R]\), with \(\mu>0\) representing the small dispersion in valuations. In a second-price auction in which each validator incurs a fixed participation cost \(c>0\), the expected surplus for the winning validator is
\[
E[\Delta] = \frac{\mu}{n+1}.
\]
Thus, a validator has a positive incentive to participate if and only if
\[
\mu > c\,(n+1).
\]
\end{theorem}

\begin{proof}
Define \(X_i = v_i - (R-\mu)\) so that \(X_i\sim U[0,\mu]\). In a second-price auction with \(n\) bidders, the expected highest value is \(E[X_{(1)}] = \frac{n}{n+1}\mu\) and the expected second-highest value is \(E[X_{(2)}] = \frac{n-1}{n+1}\mu\). The expected surplus for the winner is therefore
\[
E[X_{(1)}-X_{(2)}] = \frac{n}{n+1}\mu - \frac{n-1}{n+1}\mu = \frac{\mu}{n+1}.
\]
Because each participating validator pays the fixed cost \(c\), a validator who wins obtains a net payoff of \(\frac{\mu}{n+1} - c\) while one who loses obtains \(-c\). Thus, participation is beneficial compared to abstention (which yields 0) if and only if
\[
\frac{\mu}{n+1} - c > 0 \quad \Longleftrightarrow \quad \mu > c\,(n+1).
\]
\(\qed\)
\end{proof}

Next, consider the situation where \(k\) validators have already participated in the auction. For an additional validator joining, the auction becomes one among \(k+1\) bidders. In this case, the condition for participation is modified to
\[
\frac{\mu}{k+2} > c \quad \Longleftrightarrow \quad \mu > c\,(k+2).
\]
This corollary illustrates that as more validators join, the incremental benefit for an additional validator decreases, but participation remains attractive as long as the dispersion \(\mu\) is sufficiently large relative to the total cost incurred by all bidders.

\begin{corollary}
If \(k\) validators have already participated in the auction, an additional validator will have a positive incentive to join if and only if
\[
\mu > c\,(k+2).
\]
\end{corollary}

\subsection{Application to the Secondary Auction in a Single-Challenge Scenario}

In realistic deployments of optimistic rollups, it is often observed that only a single dispute challenge is initially submitted. Without further intervention, additional validators may refrain from joining because the early mover appears to secure the reward, leaving later participants with a guaranteed loss of the fixed cost \(c\). To overcome this free-riding behavior, we propose the deployment of a secondary auction contract that is triggered once the first challenge is registered. This contract compels additional validators to decide whether to participate in the auction.

Under the secondary auction, an additional validator now considers joining an auction with \(k\) current participants. As shown above, the participation condition is given by
\[
\mu > c\,(k+2).
\]
For instance, if initially only one challenge exists (i.e., \(k=1\)), then an extra validator will join if and only if
\[
\mu > 3c.
\]
If this condition is satisfied, the secondary auction effectively transforms the environment into one with multiple bidders, thereby enabling competitive bidding. As more validators join, the competitive pressure increases, and the equilibrium outcome will impose a cost on the malicious proposer that approximates the winning bid. In turn, the malicious proposer must incur a cost that reflects the aggregate competitive surplus, rather than merely a minimal challenge fee.

In practice, even if the valuations are nearly identical (i.e., all validators have valuations in \([R-\mu, R]\) with a small \(\mu\)), the auction mechanism’s ability to induce additional participation depends critically on the relation between \(\mu\) and the fixed cost \(c\). If \(\mu\) is sufficiently large relative to \(c\) (specifically, if \(\mu > c\,(k+2)\) for the given number \(k\) of current participants), then additional validators are incentivized to join the auction. This, in turn, increases the cost imposed on the malicious proposer when the final dispute is resolved.

Based on the analysis of the validators' profit, we can calculate the net profit of the malicious proposer with this strategy as the below corollary. It indicates, ironically, when the reward parameter for validators $\alpha$ increases and the number of validators join the dispute challneges increases, the loss of the malicious proposer decreases.

\begin{corollary}
Assuming the malicious proposer forfeits both $D_P$ and $D_G$ upon losing the dispute, and the auction operates ideally according to the Vickrey mechanism with $n$ participating validators whose valuations are drawn from $[R-\mu, R]$ where $R = \alpha D_P + D_G$, the expected net loss of the malicious proposer is given by:
\[
\text{Expected Net Loss}  = (1-\alpha)D_P + \frac{2\mu}{n+1}.
\]
\end{corollary}

\begin{proof}
Let $\pi_{\text{proposer}}$ be the expected net profit. The proposer receives the expected second-highest bid, $E[v_{(2)}]$, from the auction winner, while forfeiting deposits $D_P$ and $D_G$.
From the analysis of the second-price auction with valuations $v_i \sim U[R-\mu, R]$, where $R = \alpha D_P + D_G$:
\[
E[v_{(2)}] = R - \frac{2\mu}{n+1}
\]
The expected net profit is the revenue minus the costs:
\[
\pi_{\text{proposer}} = E[v_{(2)}] - (D_P + D_G)
\]
Substituting $E[v_{(2)}]$ and $R$:
\[
\pi_{\text{proposer}} = \left( (\alpha D_P + D_G) - \frac{2\mu}{n+1} \right) - (D_P + D_G)
\]
\[
= (\alpha - 1)D_P - \frac{2\mu}{n+1}
\]
Since $\alpha \leq 1$ and $\mu > 0$, $\pi_{\text{proposer}} \leq 0$. The expected net loss is $-\pi_{\text{proposer}}$:
\[
\text{Expected Net Loss} = (1 - \alpha)D_P + \frac{2\mu}{n+1}.
\]
\vspace{-1em} %
\end{proof}

In summary, our analysis under a second-price auction framework with slight valuation heterogeneity demonstrates that validators are incentivized to participate when the potential gain from winning outweighs the participation cost relative to the number of participants ($n$). Moreover, our findings reveal that a malicious proposer can strategically reduce their net loss by employing this auction. Increased competition, driven by either a higher validator reward parameter ($\alpha$) approaching 1 or a larger number of participating validators ($n$), allows the proposer to recoup a greater portion of their forfeited deposit through the auction proceeds (specifically, the second-highest bid which approaches the full reward $R$ as $n$ increases or $\mu$ decreases). Consequently, this increased competition leads to a decrease in the proposer's overall net loss, approaching the theoretical minimum loss of $(1-\alpha)D_P$ as $n \to \infty$ or $\mu \to 0$.

\section{Potential Solutions} \label{section: solution}

In this section, we propose two solutions to mitigate the structural vulnerability which deincentivizes optimistic rollup validators.

\subsection{Escrowed Reward Mechanism}

The first approach addresses the vulnerability arising from early challenge exploitation by modifying the reward distribution mechanism. Rather than immediately allocating the prize to the first valid challenge—as many existing protocols do—this mechanism locks the reward in escrow as soon as a valid dispute challenge is raised. As shown in Algorithm~\ref{alg:reward_lock}, the reward remains locked until \emph{all} challenges pertaining to the block are finalized. This delay in reward allocation prevents a malicious proposer from benefiting by simply finalizing the belately triggered challenge first, ensuring that the final reward distribution reflects the order and promptness of validators.

\begin{algorithm}[htb]
\caption{Escrowed Reward Mechanism for Dispute}
\label{alg:reward_lock}
\begin{algorithmic}[1]
\State \textbf{Input:} disputeStart, disputeDuration
\State \textbf{Initialize:} 
\State \quad rewardsLocked $\gets$ false, rewardPool $\gets$ 0
\State \quad challengerDeposits $\gets$ empty mapping
\Procedure{InitiateChallenge}{challenger, depositAmount}
    \If {currentTime \(<\) disputeStart + disputeDuration}
        \State challengerDeposits[challenger] $\gets$ challengerDeposits[challenger] + depositAmount
        \If {rewardsLocked is false}
            \State rewardsLocked $\gets$ true
            \State rewardPool $\gets$ current contract balance
        \EndIf
    \EndIf
\EndProcedure
\Procedure{FinalizeDispute}{}
    \If {currentTime \(\geq\) disputeStart + disputeDuration}
        \State winningChallenger $\gets$ \Call{SelectWinningChallenger}{}
        \State Transfer rewardPool to winningChallenger
        \State rewardsLocked $\gets$ false
    \EndIf
\EndProcedure
\Procedure{SelectWinningChallenger}{}
    \Comment{Determine winning challenger based on predefined criteria}
    \State \Return chosen challenger address
\EndProcedure
\end{algorithmic}
\end{algorithm}

\paragraph{Mechanism Details:}
\begin{itemize}
    \item \textbf{Escrow of Rewards:} Regardless of the order in which challenges are initiated, the entire reward (i.e., the forfeited deposits) is held in escrow until the dispute resolution period for the block ends.
    \item \textbf{Final Distribution:} Once the dispute window closes, the reward is distributed to the winning challenger. If multiple challenges are valid, the reward is apportioned or allocated according to a predefined rule.
\end{itemize}

\subsubsection{Trade-offs.}
This mechanism effectively prevents a malicious proposer from preempting the dispute process by quickly triggering a challenge using controlled validator accounts. However, a potential downside is the introduction of Miner Extractable Value (MEV) opportunities. An attacker might monitor the mempool for early challenge transactions and submit parallel challenges via MEV-boost techniques to capture a larger portion of the reward. Mitigating this risk may require validators to use private mempools or other techniques to hide their challenge submissions.

\subsection{Commit-Reveal Protocol for Challenge Decisions}

The second approach presented in Algorithm \ref{alg:commit_reveal} introduces a commit-reveal scheme that obscures validators' intentions to challenge a block. Instead of immediately broadcasting a challenge, validators submit a commitment during a designated epoch and then reveal their decision after a short delay.

\begin{algorithm}[htb]
\caption{Commit-Reveal Dispute Mechanism}
\label{alg:commit_reveal}
\begin{algorithmic}[1]
\State \textbf{Data:} Mapping \texttt{Commits} (validator \(\to\) \((commitHash, revealed, decision, blockNumber)\)), commitDeadline, revealDeadline
\Procedure{CommitChallenge}{validator, commitHash}
    \If {currentTime \(<\) commitDeadline}
        \State \texttt{Commits[validator]} \(\gets (commitHash, \; false, \; \_, \; \_)\)
    \EndIf
\EndProcedure
\Procedure{RevealChallenge}{validator, decision, blockNumber, nonce}
    \If {currentTime \(\geq\) commitDeadline and currentTime \(<\) revealDeadline}
        \If {\(\text{hash}(decision, blockNumber, nonce, validator) = \texttt{Commits[validator].commitHash}\)}
            \State Update \texttt{Commits[validator]} to \((\_,\; true,\; decision,\; blockNumber)\)
        \EndIf
    \EndIf
\EndProcedure
\Procedure{ProcessChallenges}{}
    \If {currentTime \(\geq\) revealDeadline}
        \State Process valid challenges and distribute rewards.
    \EndIf
\EndProcedure
\end{algorithmic}
\end{algorithm}

\paragraph{Mechanism Details:}
\begin{itemize}
    \item \textbf{Commit Phase:} In each epoch, validators submit a hash of their decision (e.g., \texttt{challenge} or \texttt{not challenge}), the block number, and a nonce, all signed with their private key. This commit prevents others from knowing the validator's decision in advance. Especially, nonce ensures that other validators cannot deduce the challenge status from the hash value.
    \item \textbf{Reveal Phase:} After the commit phase, validators reveal their decision and associated data. The smart contract verifies that the reveal matches the commitment.
    \item This process prevents a “follow-the-leader” attack where a malicious actor could monitor challenge submissions and mimic or preempt them.
\end{itemize}

\subsubsection{Trade-offs.}
By using a commit-reveal protocol, the challenge decision of each validator is concealed until all commitments are made. This prevents any validator from strategically adjusting their decision based on others’ actions (i.e., challenge follow-up attacks). In addition, if multiple valid challenges exist in the same block, the reward can be distributed fairly. However, this approach introduces extra overhead in the form of additional Layer-1 transactions for both commit and reveal phases. To mitigate increased transaction fees, one might consider extending the epoch duration, aggregating commits off-chain, or leveraging data blob solutions. On the other hand, while the commit phase obscures the specific timing and content of the challenge, the act of committing itself might signal an intent to challenge a particular block. To mitigate this issue, we can consider periodical attention tests using this commitment-based challenges.

In summary, both proposed solutions aim to eliminate the possibility for a malicious proposer to exploit the current dispute game incentives: The \textit{Escrowed Reward Mechanism} ensures that rewards are only distributed after the dispute game for a given block has fully resolved, thereby removing the advantage of being the first to challenge. However, it may introduce MEV risks. The \textit{Commit-Reveal Protocol} obscures validators' decisions during the dispute phase, preventing strategic follow-up challenges. The main trade-off here is the increased overhead due to additional transactions.

\section{Discussion and Concluding Remarks} \label{section: conclusion}

This work identified a critical incentive misalignment in Optimistic Rollup dispute games where validators may lack sufficient economic motivation to challenge invalid blocks due to potentially negligible net rewards. This vulnerability undermines the core security assumption relying on rational challengers.

We demonstrated how a malicious proposer can exploit this weakness, not only by leveraging controlled validators but also by strategically deploying secondary mechanisms like auctions. Such auctions, while seemingly fostering participation, can paradoxically allow the proposer to minimize their own financial penalties by capturing value from the induced validator competition. Our game-theoretic analysis revealed that increased validator engagement, whether through higher reward rates or more participants, can counterintuitively benefit the malicious actor by increasing the portion of the forfeited stake they recoup. This highlights the nuanced and sometimes non-monotonic relationship between participation incentives and overall system security. While our model provides a formal framework for understanding these dynamics, we acknowledge that real-world scenarios involve greater heterogeneity in validator valuations, costs, and risk preferences. Nonetheless, the core vulnerability stemming from the potential gap between gross rewards and net validator profits remains a significant concern.

To mitigate these risks, we proposed countermeasures focusing on decoupling reward timing from challenge finalization (escrowed reward) and obscuring validator challenge intentions (commit-reveal protocol). These solutions aim to restore robust validator incentives, though they introduce their own trade-offs regarding potential MEV opportunities and increased transactional overhead, respectively. Further research is needed to refine these mechanisms and potentially explore hybrid approaches or privacy-enhancing techniques.

\section*{Acknowledgement}

I would like to thank Junwoo Choi of Samsung SDS and Hojung Yang of Korea University for their assessment of the initial idea. I am also grateful to HyukSang Jo of Tokamak Network for his support clarifying the dispute game contract and its operation.

\bibliographystyle{splncs04}
\bibliography{references}

\end{document}